\documentclass{llncs}

\frontmatter          
\mainmatter              

\title{Speeding up shortest path algorithms}

\author{%
  Andrej Brodnik\inst{1}~\inst{2} \and
  Marko Grgurovi\v{c}\inst{1}%
}

\authorrunning{Andrej Brodnik, Marko Grgurovi\v{c}} 
\tocauthor{Andrej Brodnik, Marko Grgurovi\v{c}}
\institute{University of Primorska, Department of Information Science and Technology, Slovenia, \\
\email{andrej.brodnik@upr.si, marko.grgurovic@student.upr.si}
\and
University of Ljubljana, Faculty of Computer and Information Science, Slovenia,
}

\usepackage{makeidx}
\usepackage{fancyhdr,appendix,amssymb, latexsym,enumerate,hyperref,tkz-graph,algpseudocode, float, algorithm, graphicx, longtable, verbatim}
\usepackage[english]{babel}
\usetikzlibrary{arrows}

\begin{document}
\maketitle
\begin{abstract}
Given an arbitrary, non-negatively weighted, directed graph $G=(V,E)$ we present an algorithm that computes all pairs shortest paths in time $\mathcal{O}(m^* n + m \lg n + nT_\psi(m^*, n))$, where $m^*$ is the number of different edges contained in shortest paths and $T_\psi(m^*, n)$ is a running time of an algorithm to solve a single-source shortest path problem (SSSP). This is a substantial improvement over a trivial $n$ times application of $\psi$ that runs in $\mathcal{O}(nT_\psi(m,n))$. In our algorithm we use $\psi$ as a black box and hence any improvement on $\psi$ results also in improvement of our algorithm.

Furthermore, a combination of our method, Johnson's reweighting technique and topological sorting results in an $\mathcal{O}(m^*n + m \lg n)$ all-pairs shortest path algorithm for arbitrarily-weighted directed acyclic graphs.

In addition, we also point out a connection between the complexity of a certain sorting problem defined on shortest paths and SSSP.
\keywords{all pairs shortest path, single source shortest path}

\end{abstract}
\section{Introduction}
Let $G=(V,E)$ denote a directed graph where $E$ is the set of edges and $V$ is the set of vertices of the graph and let $\ell(\cdot)$ be a function mapping each edge to its length. Without loss of generality, we assume $G$ is strongly connected. To simplify notation, we define $m = |E|$ and $n = |V|$. Furthermore, we define $d(u,v)$ for two vertices $u,v \in V$ as the length of the shortest path from $u$ to $v$. A classic problem in algorithmic graph theory is to find shortest paths. Two of the most common variants of the problem are the single-source shortest path (SSSP) problem and the all-pairs shortest path problem (APSP). In the SSSP variant, we are asked to find the path with the least total length from a fixed vertex $s \in V$ to every other vertex in the graph. Similarly, the APSP problem asks for the shortest path between every pair of vertices $u,v \in V$. A common simplification of the problem constrains the edge length function to be non-negative, i.e. $\ell: E \rightarrow \mathbb{R}^+$, which we assume throughout the rest of the paper, except where explicitly stated otherwise. Additionally, we define $\forall (u,v) \notin E: \ell(u,v) = \infty$.

It is obvious that the APSP problem can be solved by $n$ calls to an SSSP algorithm. Let us denote the SSSP algorithm as $\psi$. We can quantify the asymptotic time bound of such an APSP algorithm as $\mathcal{O}(nT_{\psi}(m,n))$ and the asymptotic space bound as $\mathcal{O}(S_{\psi}(m,n))$, where $T_{\psi}(m,n)$ is the time required by algorithm $\psi$ and $S_{\psi}(m,n)$ is the space requirement of the same algorithm. We assume that the time and space bounds can be written as functions of $m$ and $n$ only, even though this is not necessarily the case in more ``exotic'' algorithms that depend on other parameters of $G$. Note, that if we are required to store the computed distance matrix, then we will need at least $\Theta(n^2)$ additional space. If we account for this, then the space bound becomes $\mathcal{O}(S_{\psi}(m,n) + n^2)$.

In this paper we are interested in the following problem: what is the best way to make use of an SSSP algorithm $\psi$ when solving APSP? There exists some prior work on a very similar subject in the form of an algorithm named the Hidden Paths Algorithm~\cite{Karger+al:93}. The Hidden Paths Algorithm is essentially a modification of Dijkstra's algorithm~\cite{dijkstra} to make it more efficient when solving APSP. Solving the APSP problem by repeated calls to Dijkstra's algorithm requires $\mathcal{O}(mn + n^2 \lg n)$ time using Fibonacci heaps~\cite{fib}. The Hidden Paths Algorithm then reduces the running time to $\mathcal{O}(m^*n + n^2 \lg n)$. The quantity $m^*$ represents the number of edges $(u,v) \in E$ such that $(u,v)$ is included in at least one shortest path. In the Hidden Paths Algorithm this is accomplished by modifying Dijkstra's algorithm, so that it essentially runs in parallel from all vertex sources in $G$, and then reusing the computations performed by other vertices. The idea is simple: we can delay the inclusion of an edge $(u,v)$ as a candidate for forming shortest paths until vertex $u$ has found $(u,v)$ to be the shortest path to $v$. However, the Hidden Paths Algorithm is limited to Dijkstra's algorithm, since it explicitly sorts the shortest path lists by path lengths, through the use of a priority queue. As a related algorithm, we also point out that a different measure $|UP|$ related to the number of so-called uniform paths has also been exploited to yield faster algorithms~\cite{DemetrescuI06}.

In Sections~\ref{sec:algo},~\ref{sec:improv} and~\ref{sec:dags} we show that there is a method for solving APSP which produces the shortest path lists of individual vertices in sorted order according to the path lengths. The interesting part is that it can accomplish this without the use of priority queues of any form and requires only an SSSP algorithm to be provided. This avoidance of priority queues permits us to state a time complexity relationship between a sorted variant of APSP and SSSP. Since it is very difficult to prove meaningful lower bounds for SSSP, we believe this connection might prove useful.

As a direct application of our approach, we show that an algorithm with a similar time bound to the Hidden Paths Algorithm can be obtained. Unlike the Hidden Paths Algorithm, the resulting method is general in that it works for any SSSP algorithm, effectively providing a speed-up for arbitrary SSSP algorithms. The proposed method, given an SSSP algorithm $\psi$, has an asymptotic worst-case running time of $\mathcal{O}(m^* n + m \lg n + nT_{\psi}(m^*, n))$ and space $\mathcal{O}(S_{\psi}(m,n) + n^2)$. We point out that the $m^*n$ term is dominated by the $nT_{\psi}(m^*, n)$ term, but we feel that stating the complexity in this (redundant) form makes the result clearer to the reader. For the case of $\psi$ being Dijkstra's algorithm, this is asymptotically equivalent to the Hidden Paths Algorithm. However, since the algorithm $\psi$ is arbitrary, we show that the combination of our method, Johnson's reweighting technique~\cite{johnson} and topological sorting gives an $\mathcal{O}(m^*n + m \lg n)$ APSP algorithm for arbitrarily-weighted directed acyclic graphs.

\section{Preliminaries}
Throughout the paper and without loss of generality, we assume that we are not interested in paths beginning in $v$ and returning back to $v$. We have previously defined the edge length function $\ell(\cdot)$, which we now extend to the case of paths. Thus, for a path $\pi$, we write $\ell(\pi)$ to denote its length, which corresponds to the sum of the length of its edges.

Similar to the way shortest paths are discovered in Dijkstra's algorithm, we rank shortest paths in nondecreasing order of their lengths. Thus, we call a path $\pi$ the $k$-th shortest path if it is at position $k$ in the length-sorted shortest path list. The list of paths is typically taken to be from a single source to variable target vertices. In contrast, we store paths from variable sources to a single target. By reversing the edge directions we obtain the same lists, but it is conceptually simpler to consider the modified case. Thus, the $k$-th shortest path of vertex $v$ actually represents the $k$-th shortest incoming path into $v$. We will now prove a theorem on the structure of shortest paths, which is the cornerstone of the proposed algorithm.
\begin{definition}
\label{def:ordpath}
\textbf{(Ordered shortest path list $P_v$)}\\
Let $P_v = ( \pi_1, \pi_2, ..., \pi_{n-1} )$ denote the shortest path list for each vertex $v \in V$. Then, let $P_{v,k}$ denote the $k$-th element in the list $P_v$. The shortest path lists are ordered according to path lengths, thus we have $\forall i,j: 0 < i < j < n \Rightarrow \ell(\pi_i) \leq \ell(\pi_j)$.
\end{definition}
\begin{theorem}
\label{thm:sortedpath}
To determine $P_{v,k}$ we only need to know every edge $\{(u,v) \in E$ $|$ $\forall u \in V\}$ and the first $k$ elements of each list $P_u$, where $(u,v) \in E$.
\end{theorem}
\begin{proof}
We assume that we have found the first $k$ shortest paths for all neighbors of $v$, and are now looking for the $k$-th shortest path into $v$, which we denote as $\pi_{k}$. There are two possibilities: either $\pi_{k}$ is simply an edge $(u,v)$, in which case we already have the relevant information, or it is the concatenation of some path $\pi$ and an edge $(u,v)$. The next step is to show that $\pi$ is already contained in $P_{u,i}$ where $i \leq k$.

We will prove this by contradiction. Assume the contrary, that $\pi$ is either not included in $P_u$, or is included at position $i>k$. This would imply the existence of some path $\pi'$ for which $\ell(\pi') \leq \ell(\pi)$ and which is contained in $P_u$ at position $i \leq k$. Then we could simply take $\pi_{k}$ to be the concatenation of $(u,v)$ and $\pi'$, thereby obtaining a shorter path than the concatenation of $(u,v)$ and $\pi$. However, this is not yet sufficient for a contradiction. Note that we may obtain a path that is shorter, but connects vertices that have an even shorter path between them, i.e. the path is not the shortest path between the source $s$ and target $v$.

To show that it does contradict our initial assumption, we point out that $P_u$ contains $k$ shortest paths, therefore it contains shortest paths from $k$ unique sources. In contrast, the list $P_{v}$ contains at most $k-1$ shortest paths. By a counting argument we have that there must exist a path $\pi'$, stored in $P_u$ with an index $i \leq k$, which originates from a source vertex $s$ that is not contained in $P_{v}$, thereby obtaining a contradiction. \qed
\end{proof}

\section{The algorithm}
\label{sec:algo}
Suppose we have an SSSP algorithm $\psi$ and we can call it using $\psi(V, E, s)$ where $V$ and $E$ correspond to the vertex and edge sets, respectively and $s$ corresponds to the source vertex. The method we propose works in the fundamental comparison-addition model and does not assume a specific kind of edge length function, except the requirement that it is non-negative. However, the algorithm $\psi$ that is invoked can be arbitrary, so if $\psi$ requires a different model or a specific length function, then implicitly by using $\psi$, our algorithm does as well. 

First we give a simpler variant of the algorithm, resulting in bounds $\mathcal{O}(mn + nT_{\psi}(m^*, n))$. We limit our interaction with $\psi$ only to execution and reading its output. To improve the running time we construct a graph $G' = (V',E')$ on which we run $\psi$. There are two processes involved: the method for solving APSP which runs on $G$, and the SSSP algorithm $\psi$ which runs on $G'$. Let $n' = |V'|$ and $m' = |E'|$. We will maintain $m' \leq m^* + n$ and $n' = n + 1$ throughout the execution. There are $n-1$ phases of the main algorithm, each composed of three steps: (1) Prepare the graph $G'$; (2) Run $\psi$ on $G'$; and (3) Interpret the results of $\psi$.

Although the proposed algorithm effectively works on $n-1$ new graphs, these graphs are similar to one another. Thus, we can consider the algorithm to work only on a single graph $G'$, with the ability to modify edge lengths and introduce new edges into $G'$. Initially we define $V' = V \cup \{i\}$, where $i$ is a new vertex unrelated to the graph $G$. We create $n$ new edges from $i$ to every vertex $v \in V$, i.e. $E' = \bigcup_{v \in V} \{(i,v)\}$. We set the cost of these edges to some arbitrary value in the beginning.

\begin{definition} \emph{(Shortest path list for vertex $v$, $S_v$)}
\label{def:splist}
The shortest path list of some vertex $v \in V$ is denoted by $S_v$. The length of $S_v$ is at most $n+1$ and contains pairs of the form $(a,\delta)$ where $a \in V \cup \{null\}$ and $\delta \in \mathbb{R^+}$. The first element of $S_v$ is always $(v,0)$, the last element plays the role of a sentinel and is always $(null, \infty)$. For all inner (between the first and the last element) elements $(a,\delta)$, we require that $\delta=d(a,v)$. A list with $k\leq n-1$ inner elements:
\[S_v = \big( (v,0), (a_1, \delta_1), (a_2, \delta_2),...,(a_k, \delta_k), (null, \infty) \big).\]
\end{definition}

Next we describe the data structures. Each vertex $v \in V$ keeps its shortest path list $S_v$, which initially contains only two pairs $(v,0)$ and $(null, \infty)$. For each edge $(u,v) \in E$, vertex $v$ keeps a pointer $p[(u,v)]$, which points to some element in the shortest path list $S_u$. Initially, each such pointer $p[(u,v)]$ is set to point to the first element of $S_u$.

\begin{definition} \emph{(Viable pair for vertex $v$)}
\label{def:viabpair}
A pair $(a,\delta)$ is viable for a vertex $v \in V$ if $\forall (a',\delta') \in S_v: a \neq a'$. Alternatively, if $a = null$ we define the pair as viable.
\end{definition}

\begin{definition} \emph{(Currently best pair for vertex $v$, $(a_v, \delta_v)$)}
\label{def:bestpair}
A pair $(a_v,\delta_v) \in S_w$, where $(w,v) \in E$ is the currently best pair for vertex $v$ if and only if $(a_v,\delta_v)$ is viable for $v$ and: $\forall (u,v) \in E: \forall (a',\delta') \in S_u: (a',\delta')$ viable for $v$ and $\delta'+\ell(u,v) \geq \delta_v + \ell(w,v)$.
\end{definition}

We now look at the first step taken in each phase of the algorithm: preparation of the graph $G'$. In this step, each vertex $v$ finds the currently best pair $(a_v, \delta_v)$. To determine the currently best pair, a vertex $v$ inspects the elements pointed to by its pointers $p[(u,v)]$ for each $(u,v) \in E$ in the following manner: For each pointer $p[(u,v)]$, vertex $v$ keeps moving the pointer to the next element in the list $S_u$ until it reaches a viable pair, and takes the minimum amongst these as per Definition~\ref{def:bestpair}. We call this process \textit{reloading}.

Once reloaded we modify the edges in the graph $G'$. Let $(a_v,\delta_v) \in S_w$ where $(w,v) \in E$ be the currently best pair for vertex $v$, then we set $\ell(i,v) \gets \delta_v+\ell(w,v)$. Now we call $\psi(V', E', i)$. Suppose the SSSP algorithm returns an array $\Pi[~]$ of length $n$. Let each element $\Pi[v]$ be a pair $(c,\delta)$ where $\delta$ is the length of the shortest path from $i$ to $v$, and $c$ is the first vertex encountered on this path. When determining the first vertex on the path we exclude $i$, i.e. if the path is $\pi_v = \{(i,v)\}$ then $\Pi[v].c=v$. The inclusion of the first encountered vertex is a mere convenience, and can otherwise easily be accomodated by examining the shortest path tree returned by the algorithm. For each vertex $v \in V$ we append the pair $(a_{\Pi[v].c}, \Pi[v].\delta)$ to its shortest path list. Note, that the edges $(i,v) \in E'$ are essentially shorthands for paths in $G$. Thus, $a_{\Pi[v].c}$ represents the source of the path in $G$. We call this process \textit{propagation}.

After propagation, we modify the graph $G'$ as follows. For each vertex $v \in V$ such that $\Pi[v].c = v$, we check whether the currently best pair $(a_v,\delta_v) \in S_u$ that was selected during the reloading phase is the first element of the list $S_u$. If it is the first element, then we add the edge $(u,v)$ into the set $E'$. This concludes the description of the algorithm. We formalize the procedure in pseudocode and obtain Algorithm~\ref{alg:apsp}. To see why the algorithm correctly computes the shortest paths, we prove the following two lemmata.
\begin{algorithm}
\caption{All-pairs shortest path}
\begin{algorithmic}[1]
\Procedure{APSP}{$V, E, \psi$}
	\State $V' := V \cup \{i\}$
	\State $E' := \bigcup_{\forall v \in V} \{(i,v)\}$
	\State $best[~] :=$ new array $[n]$ of pairs $(a,\delta)$
	\State $solved[~][~] :=$ new array $[n][n]$ of boolean values
	\State Initialize $solved[~][~]$ to $false$
	\ForAll{$v \in V$}
		\State $S_v.append($ $(v, 0)$ $)$
	\EndFor
	\For{$k:=1$ to $n-1$}
		\ForAll{$v \in V$} \Comment{Reloading}
			\State $best[v] := (null, \infty)$
			\ForAll{$u \in V$ s.t. $(u,v) \in E$}
				\While{$solved[v][p[(u,v)].a]$}
					\State $p[(u,v)].next()$ \Comment{An end-of-list element is always viable}
				\EndWhile
				\If{$p[(u,v)].\delta + \ell(u,v) < best[v].\delta$}
					\State $best[v].a := p[(u,v)].a$
					\State $best[v].\delta := p[(u,v)].\delta + \ell(u,v)$
				\EndIf
			\EndFor
			\State $\ell(i,v) := best[v].\delta$ \Comment{Considering only $k-1$ neighboring paths}
		\EndFor
		\State $\Pi[~] :=$ \Call{$\psi$}{$V',E',i$}
		\ForAll{$v \in V$} \Comment{Propagation}
			\State $S_v.append($ $(best[\Pi[v].c].a,$ $\Pi[v].\delta)$ $)$
			\State $solved[v][best[\Pi[v].c].a] := true$
			\If{$\Pi[v].c = v$ and $best[v]$ was the first element of some list $S_u$}
				\State $E' := E' \cup {(u,v)}$
			\EndIf
		\EndFor
	\EndFor
\EndProcedure
\end{algorithmic}
\label{alg:apsp}
\end{algorithm}
\begin{lemma}
\label{lem:preservation}
For each vertex $v \in V$ whose $k$-th shortest path was found during the reloading step, $\psi(V',E',i)$ finds the edge $(i,v)$ to be the shortest path into $v$.
\end{lemma}
\begin{proof}
For the case when the $k$-th shortest path depends only on a path at position $j<k$ in a neighbor's list, the path is already found during the reloading step. What has to be shown is that this is preserved after the execution of the SSSP algorithm. Consider a vertex $v \in V$ which has already found the $k$-th shortest path during the reloading step. This path is represented by the edge $(i,v)$ of the same length as the $k$-th shortest path. Now consider the case that some path, other than the edge $(i,v)$ itself, would be found to be a better path to $v$ by the SSSP algorithm. Since each of the outgoing edges of $i$ represents a path in $G$, this would mean that taking this path and adding the remaining edges used to reach $v$ would consistute a shorter path than the $k$-th shortest path of $v$. Let us denote the path obtained by this construction as $\pi'$. Clearly this is a contradiction unless $\pi'$ is not the $k$-th shortest path, i.e. a shorter path connecting the two vertices is already known.

Without loss of generality, assume that $\pi' = \{(i,u),(u,v)\}$. However, $\ell(\pi')$ can only be shorter than $\ell(i,v)$ if $v$ could not find a viable (non-$null$) pair in the list $S_u$, since otherwise a shorter path would have been chosen in the reloading phase. This means that all vertex sources (the $a$ component of a pair) contained in the list $S_u$ are also contained in the list $S_v$. Therefore a viable pair for $u$ must also be a viable pair for $v$. This concludes the proof by contradiction, since the path obtained is indeed the shortest path between the two vertices. \qed
\end{proof}

\begin{lemma}
\label{lem:kthpath}
$\psi(V',E',i)$ correctly computes the $k$-th shortest paths for all vertices $v \in V$ given only $k-1$ shortest paths for each vertex.
\end{lemma}
\begin{proof}
The case when the $k$-th path requires only $k-1$ neighboring paths to be known has already been proven by the proof of Lemma~\ref{lem:preservation}. We now consider the case when the $k$-th path depends on a neighbor's $k$-th path. If the $k$-th path of vertex $v$ requires the $k$-th path from the list of its neighbor $u$, then we know the $k$-th path of $u$ must be the same as that of $v$ except for the inclusion of the edge $(u,v)$. The same argument applies to the dependency of vertex $u$ on its neighbor's list. Thus, the path becomes shorter after each such dependency, eventually becoming dependent on a path included at position $j<k$ in a neighbor's list (this includes edges), which has already been found during the reloading step and is preserved as the shortest path due to Lemma~\ref{lem:preservation}.

We now proceed in the same way that we obtained the contradiction in the proof of Lemma~\ref{lem:preservation}, except it is not a contradiction in this case. What follows is that any path from $i$ to $v$ in $G'$ which is shorter than $\ell(i,v)$ must represent a viable pair for $v$. It is easy to see, then, that the shortest among these paths is the $k$-th shortest path for $v$ in $G$ and also the shortest path from $i$ to $v$ in $G'$. \qed
\end{proof}
\subsection{Time and space complexity}
First, we look at the time complexity. The main loop of Algorithm~\ref{alg:apsp} (lines $7$--$29$) performs $n-1$ iterations. The reloading loop (lines $8$--$20$) considers each edge $(u,v) \in E$ which takes $m$ steps. This amounts to $\mathcal{O}(mn)$. Since each shortest path list is of length $n+1$, each pointer is moved to the next element $n$ times over the execution of the algorithm. There are $m$ pointers, so this amounts to $\mathcal{O}(mn)$. Algorithm $\psi$ is executed $n-1$ times. In total, the running time of Algorithm~\ref{alg:apsp} is $\mathcal{O}(mn + nT_\psi(m^*,n))$.

The space complexity of Algorithm~\ref{alg:apsp} is as follows. Each vertex keeps track of its shortest path list, which is of size $n+1$ and amounts to $\Theta(n^2)$ space over all vertices. Since there are exactly $m$ pointers in total, the space needed for them is simply $\mathcal{O}(m)$. On top of the costs mentioned, we require as much space as is required by algorithm $\psi$. In total, the combined space complexity for Algorithm~\ref{alg:apsp} is $\mathcal{O}(n^2 + S_\psi(m^*,n))$.
\subsection{Implications}
We will show how to further improve the time complexity of the algorithm in Section~\ref{sec:improv}, but already at its current stage, the algorithm reveals an interesting relationship between the complexity of non-negative SSSP and a stricter variant of APSP.
\begin{definition}
\label{def:sapsp}
\textbf{(Sorted all-pairs shortest path $\mathit{SAPSP}$)}\\
The problem $\mathit{SAPSP}(m,n)$ is that of finding shortest paths between all pairs of vertices in a non-negatively weighted graph with $m$ edges and $n$ vertices in the form of $P_v$ for each $v \in V$ (see Definition~\ref{def:ordpath}).
\end{definition}
\begin{theorem}
\label{thm:complexity}
Let $T_{\mathit{SSSP}}$ denote the complexity of the single-source shortest path problem on non-negatively weighted graphs with $m$ edges and $n$ vertices. Then the complexity of $\mathit{SAPSP}$ is at most $\mathcal{O}(nT_{\mathit{SSSP}})$.
\end{theorem}
\begin{proof}
Given an algorithm $\psi$ which solves SSSP, we can construct a solution to SAPSP in time $\mathcal{O}( nT_\psi(m,n))$ according to Algorithm~\ref{alg:apsp}, since the lists $S_v$ found by the algorithm are ordered by increasing distance from the source. \qed
\end{proof}
What Theorem~\ref{thm:complexity} says is that when solving APSP, either we can follow in the footsteps of Dijkstra and visit vertices in increasing distance from the source without worrying about a sorting bottleneck, or that if such a sorting bottleneck exists, then it proves a non-trivial lower bound for the single-source case.
\section{Improving the time bound}
\label{sec:improv}
The algorithm presented in the previous section has a running time of $\mathcal{O}(mn + nT_\psi(m^*,n))$. We show how to bring this down to $\mathcal{O}(m^*n + m \lg n + nT_\psi(m^*,n))$. We sort each set of incoming edges $E_v = \bigcup_{(u,v) \in E} \{(u,v)\}$ by edge lengths in non-decreasing order. By using any off-the-shelf sorting algorithm, this takes $\mathcal{O}(m \lg n)$ time.

We only keep pointers $p[(u,v)]$ for the edges which are shortest paths between $u$ and $v$, and up to one additional edge per vertex for which we do not know whether it is part of a shortest path. Since edges are sorted by their lengths, a vertex $v$ can ignore an edge at position $t$ in the sorted list $E_v$ until the edge at position $t-1$ is either found to be a shortest path, or found \textit{not} to be a shortest path. For some edge $(u,v)$ the former case simply corresponds to using the first element, i.e. $u$, provided by $p[(u,v)]$ as a shortest path. The latter case on the other hand, is \textit{not} using the first element offered by $p[(u,v)]$, i.e. finding it is not viable during the reloading phase. Whenever one of these two conditions is met, we include the next edge in the sorted list as a pointer, and either throw away the previous edge if it was found not to be a shortest path, or keep it otherwise. This means the total amount of pointers is at most $m^* + n$ at any given time, which is $\mathcal{O}(m^*)$, since $m^*$ is at least $n$. The total amount of time spent by the algorithm then becomes $\mathcal{O}(m^*n + m \lg n + nT_\psi(m^*,n))$.
\begin{theorem}
\label{thm:apsp}
Let $\psi$ be an algorithm which solves the single-source shortest path problem on non-negatively weighted graphs. Then, the all-pairs shortest path problem on non-negatively weighted graphs can be solved in time $\mathcal{O}(m^*n + m \lg n + nT_\psi(m^*,n))$ and space $\mathcal{O}(n^2 + S_\psi(m^*,n))$ where $T_\psi(m,n)$ is the time required by algorithm $\psi$ on a graph with $m$ edges and $n$ nodes and $S_\psi(m,n)$ is the space required by algorithm $\psi$ on the same graph.
\end{theorem}
\begin{proof}
See discussion above and in Section~\ref{sec:algo}. \qed
\end{proof}
\section{Directed acyclic graphs}
\label{sec:dags}
A combination of a few techniques yields an $\mathcal{O}(m^*n + m \lg n)$ APSP algorithm for arbitrarily weighted directed acyclic graphs (DAGs). The first step is to transform the original (possibly negatively-weighted) graph into a non-negatively weighted graph through Johnson's~\cite{johnson} reweighting technique. Instead of using Bellman-Ford in the Johnson step, we visit nodes in their topological order, thus obtaining a non-negatively weighted graph in $\mathcal{O}(m)$ time. Next, we use the improved time bound algorithm as presented in Section~\ref{sec:improv}. For the SSSP algorithm, we again visit nodes according to their topological order. Note that if the graph $G$ is a DAG then $G'$ is also a DAG. The reasoning is simple: the only new edges introduced in $G'$ are those from $i$ to each vertex $v \in V$. But since $i$ has no incoming edges, the acyclic property of the graph is preserved. The time bounds become $\mathcal{O}(m)$ for Johnson's step and $\mathcal{O}(m^*n + m \lg n + nT_\psi(m^*,n))$ for the APSP algorithm where $T_\psi(m^*,n) = \mathcal{O}(m^*)$. Thus, the combined asymptotic running time is $\mathcal{O}(m^*n + m \lg n)$. The asymptotic space bound is simply $\Theta(n^2)$.
\begin{theorem}
\label{thm:dag}
All-pairs shortest path on directed acyclic graphs can be solved in time $\mathcal{O}(m^*n + m \lg n)$ and $\Theta(n^2)$ space.
\end{theorem}
\begin{proof}
See discussion above. \qed
\end{proof}
\section{Discussion}
In this paper we have shown that the ``standard'' approach to solving APSP via independent SSSP computations can be improved upon even if we know virtually nothing about the SSSP algorithm itself. However, we should mention that in recent years, asymptotically efficient algorithms for APSP have been formulated in the so-called component hierarchy framework. These algorithms can be seen as computing either SSSP or APSP. Our algorithm is only capable of speeding up SSSP hierarchy algorithms, such as Thorup's~\cite{thorup_sssp}, but not those which reuse the hierarchy, such as Pettie's~\cite{pettie_apsp_real}, Pettie-Ramachandran~\cite{pettie_rama} or Hagerup's~\cite{hagerup} since our SSSP reduction requires modifications to the graph $G'$. These modifications would require the hierarchy to be recomputed, making the algorithms prohibitively slow. This raises the following question: is there a way to avoid recomputing the hierarchy at each step, while keeping the number of edges in the hierarchy $\mathcal{O}(m^*)$?

Further, if there exists an $o(mn)$ algorithm for the arbitrarily-weighted SSSP problem, then by using Johnson's reweighting technique, our algorithm might become an attractive solution for that case. For the general case, no such algorithms are known, but for certain types of graphs, there exist algorithms with an $o(mn)$ asymptotic time bound~\cite{goldberg,gabowtarscaling}.

Furthermore, we can generalize the approach used on DAGs. Namely, in Algorithm~\ref{alg:apsp} we can use an SSSP algorithm $\psi$ that works on a specialized graph $G$, as long our constructed graph $G'$ has these properties. Therefore, our algorithm can be applied to undirected graphs, integer-weighted graphs, etc., but it cannot be applied, for example, to planar graphs, since $G'$ is not necessarily planar.

Finally, we have shown a connection between the sorted all-pairs shortest path problem and the single-source shortest path problem. If a meaningful lower bound can be proven for SAPSP, then this would imply a non-trivial lower bound for SSSP. Alternatively, if SAPSP can be solved in $O(mn)$ time, then this implies a Dijkstra-like algorithm for APSP, which visits vertices in increasing distance from the source.
\bibliographystyle{splncs}
\bibliography{biblio}

\end{document}